\documentclass[runningheads,envcountsame]{llncs}
\usepackage[T1]{fontenc}
\usepackage{amsmath,amssymb}
\usepackage{mathtools}
\usepackage{enumerate}
\usepackage{hyperref}
\usepackage[noabbrev,capitalize]{cleveref}
\Crefname{lemma}{Lemma}{Lemmata}
\usepackage{ebproof}
\usepackage{xspace}

\newcommand{\qedhere}{\qed}

\renewcommand{\emptyset}{\varnothing}

\newcommand{\m}[1]{\mathsf{#1}}
\newcommand{\mi}[1]{\mathit{#1}}
\newcommand{\seq}[2][n]{{#2_1},\ldots,{#2_{#1}}}
\newcommand{\aseq}[2][n]{{#2_1}\cdots{#2_{#1}}}

\newcommand{\bN}{\mathbb{N}}
\newcommand{\cF}{\mathcal{F}}

\newcommand{\cM}{\mathcal{M}}
\newcommand{\cR}{\mathcal{R}}

\newcommand{\cV}{\mathcal{V}}

\newcommand{\Var}{\mathcal{V}\mathsf{ar}}
\newcommand{\HV}{\mathcal{HV}}

\renewcommand{\Sigma}{\cF_0}
\newcommand{\SN}{{\mathsf{SN}}}
\newcommand{\mul}{\mathsf{mul}}
\newcommand{\mpo}{\mathsf{mpo}}
\newcommand{\ampo}{\mathsf{ampo}}
\newcommand{\app}{\mathop{\circ}}
\newcommand{\APP}{@}
\newcommand{\reify}[1]{\hat{#1}}
\newcommand{\gtB}{\mathrel{\triangleright}}
\newcommand{\gtE}{\mathrel{\triangleright}_\eta}
\newcommand{\gtM}{>_\ampo}

\newcommand{\AMPO}[1]{\textbf{\labelcref{ampo#1}}}
\newcommand{\MPO}[1]{\labelcref{mpo#1}}
\newcommand{\aaR}{\mathsf{aa}_\cR}

\begin{document}
\title{An Applicative Multiset Path Order \\ (Extended Version)}
\author{Nao~Hirokawa\orcidID{0000-0002-8499-0501}
\and Teppei~Saito\orcidID{0009-0001-9786-0044}
\and Teppei~Tanaka
\and Wataru~Yachi}
\authorrunning{N. Hirokawa et al.}
\institute{JAIST, Nomi, Japan \email{\{hirokawa,saito\}@jaist.ac.jp}}
\maketitle
\begin{abstract}
We present a variant of the multiset path order for untyped applicative
term rewriting.  Compared to existing work, our variant incorporates
two distinctive features, dubbed arity assignment and reification,
to overcome difficulties in handling partial and variable
application.
\end{abstract}

\section{Introduction}

Dershowitz's multiset path order~\cite{D82} (originally called recursive path order)
is a popular reduction order for proving termination of first-order term rewriting.
This paper presents an extension of the order to 
untyped \emph{applicative} term rewriting~\cite{KKSdeF98}
which allows us to model lambda-free higher-order functions with first-order terms.
For example, the higher-order function $\m{map}$ is represented by:
\begin{align*}
\m{map}~f~\m{nil} &\to \m{nil}
&
\m{map}~f~(\m{cons}~x~\mi{xs}) &\to \m{cons}~(f~x)~(\m{map}~f~\mi{xs})
\end{align*}
This expressiveness comes at the price of difficulty in proving termination.
As a matter of the fact, the multiset path order
(or more generally, any reduction order with the subterm property)
cannot even show the termination of the one-rule rewrite system
$\m{mult}~(\m{s}~x)~y \to \m{add}~(\m{mult}~x~y)~y$.

In fact, variants of the recursive path order~\cite{B21,BWW17} for applicative terms have 
already been investigated 
in the context of higher-order theorem proving.
Our work follows this line of research.
For instance, the Knuth--Bendix completion~\cite{KB70} 
with our order
can transform the equations
\begin{align*}
\m{append}~\m{nil}~\mi{ys} &\approx \mi{ys}
&
\m{append}~(\m{cons}~x~\mi{xs})~\mi{ys}
&\approx \m{cons}~x~(\m{append}~\mi{xs}~\mi{ys})
&
\m{id}~x &\approx x
\\
\m{map}~f~\m{nil} &\approx \m{nil}
&
\m{map}~f~(\m{cons}~x~\m{nil})
&\approx \m{cons}~(f~x)~\m{nil}
\\
\m{map}~\m{id} &\approx \m{id}
&
\m{map}~f~(\m{append}~\mi{xs}~\mi{ys})
&\approx \makebox[0mm][l]{$\m{append}~(\m{map}~f~\mi{xs})~(\m{map}~f~\mi{ys})$}
\end{align*}
into a complete (terminating and confluent) rewrite system
(see \cref{ex:map}).

Technically, our variant has two distinctive features: \emph{arity assignment}
for constants and \emph{reification} of function application.  The former is
used for simulating the behavior of the original multiset path order,
while the latter, inspired by uncurrying~\cite{GTS05,HMZ13,KKSdeF98},
is used for handling partial
application and variable application like $\m{map}~f$ and $f~x$.
Our order forms a reduction order capable of
showing termination of typical higher-order functions such as $\m{map}$ and
$\m{filter}$.
Moreover, it satisfies the \emph{gracefulness} property which 
states a form of correspondence between Dershowitz's multiset path order and ours
and thereby provides a minimum guarantee in termination proving.

The remaining part of the paper is organized as follows.
\cref{sec:prelim} serves as preliminaries.  We build our order in three
steps: in \cref{sec:base}  we define a base order which is shown transitive
and well-founded; in \cref{sec:mono} 
it is lifted to a well-founded order that is closed under contexts;
and finally in \cref{sec:stability} it is further lifted to a reduction
order by a folklore technique. The resulting order is
evaluated in \cref{sec:evaluation}.  We conclude by discussing related
work in \cref{sec:conclusion}.

This paper is an extended version of our conference paper~\cite{IJCAR} with omitted proofs in the appendix.

\section{Preliminaries}
\label{sec:prelim}

Throughout the paper, we assume familiarity with term
rewriting~\cite{BN98,TeReSe}.  Below, we briefly recall
Dershowitz's multiset path order and applicative terms.

Let $\cV$ be an infinite set of variables, and $\cF$ a set of function
symbols with $\cV \cap \cF = \varnothing$.  Terms are inductively
defined as follows: all variables are terms; and if $f$ is an $n$-ary
function symbol in $\cF$ and $\seq{t}$ are terms then $f(\seq{t})$ is a
term.  We simply write $f$ for $f()$.
\emph{Ground terms} are terms that do not contain variables.
A \emph{substitution} $\sigma$ is a mapping from variables to terms such
that $\sigma(x) \neq x$ for only finitely many variables $x$.  In
particular, we write $\{ x \mapsto t \}$ for the substitution that maps $x$
to $t$ and other variables $y$ to $y$ itself.  Given a term $t$ and a substitution
$\sigma$, we write $t\sigma$ for the term obtained by replacing every
variable $x$ occurring in $t$ with $\sigma(x)$.  A \emph{context} $C$ is a
term containing exactly one \emph{hole} $\square$, which is a
constant not in $\cF$. The term obtained by replacing $\square$
in $C$ by a term $t$ is denoted by $C[t]$.  
\emph{Ground contexts} are contexts that do not contain variables.
The set of variables occurring in a term $t$ is written as $\Var(t)$.  
A \emph{rewrite rule} $\ell \to r$
is a pair $(\ell, r)$ of terms such that $\ell$ is not a variable and
$\Var(\ell) \supseteq \Var(r)$.
A set $\cR$ of rewrite rules is called
a \emph{term rewrite system} (TRS).  We write $s \to_\cR t$ if 
$s = C[\ell\sigma]$ and $t = C[r\sigma]$ for some rewrite rule
$\ell \to r \in \cR$, context $C$, and substitution $\sigma$.

A TRS $\cR$ is \emph{terminating} if there is no infinite sequence
$t_0 \to_\cR t_1 \to_\cR \cdots$ with respect to $\to_\cR$.
Alternatively, it means that the relation $\to_\cR$ is well-founded.
A basic method for proving termination is \emph{reduction orders}, which
are well-founded orders on terms that is closed under substitutions and
contexts.  Here $>$ is said to be \emph{closed under substitutions} if 
$s \sigma > t\sigma$ whenever $s > t$ and $\sigma$ is a substitution.
Similarly, $>$ is said to be \emph{closed under (ground) contexts} 
if $C[s] > C[t]$ whenever $s > t$ and $C$ is a (ground) context,
respectively.
Reduction orders characterize termination of term rewrite systems.

\begin{proposition}
\label{prop:compatibility}
A TRS $\cR$ is terminating if and only if there exists a
reduction order $>$ with $\cR \subseteq {>}$.
\end{proposition}

A key ingredient of the multiset path order is multiset extension.
Let $>$ be a binary relation on a set $A$.
We write $\cM(A)$ for the set of all finite multisets over $A$.
The \emph{multiset extension} $>^\mul$ of $>$
is defined on $\cM(A)$ as follows: $M >^\mul N$ if there exist multisets 
$X, Y \in \cM(A)$ such that $N = (M - X) \uplus Y$,
$\varnothing \neq X \subseteq M$, and every $b \in Y$ admits an element 
$a \in X$ with $a > b$.  It is known that $>^\mul$ is a well-founded order
if $>$ is a well-founded order~\cite{DM79}.

Now we recall the multiset path order (MPO for short)~\cite{D82}.
A \emph{precedence} $\succsim$ on a signature $\cF$ is a quasi-order on
$\cF$. We denote the strict part $\succsim \setminus \precsim$
of $\succsim$ by $\succ$, and the equivalence part $\succsim \cap \precsim$
by $\sim$.  We call $\succsim$ well-founded if $\succ$ is.
The reflexive closure and the reflexive and transitive
closure of a relation $\leadsto$ are denoted by $\leadsto^=$ and
$\leadsto^*$, respectively.

\begin{definition}
\label{def:mpo}
Let $\succsim$ be a precedence.  The \emph{multiset path order}
$>_\mpo$ is inductively defined on terms as follows: $s >_\mpo t$ if 
$s = f(\seq[m]{s})$ and one of the following conditions holds.
\begin{enumerate}
\item
\label{mpo1}
$s_i >_\mpo^= t$ for some $i$.
\item
\label{mpo2}
$t = g(\seq{t})$, $f \succ g$, and
$s >_\mpo t_j$ for all $j$.
\item
\label{mpo3}
$t = g(\seq{t})$, $f \succsim g$, and
$\{\seq[m]{s}\} >_\mpo^\mul \{\seq{t}\}$.
\end{enumerate}
\end{definition}

\begin{theorem}[\textnormal{\cite{D82}}]
\label{thm:mpo}
The relation $>_\mpo$ is a reduction order if the precedence $\succsim$ is well-founded.
\end{theorem}

It is worth noting that the MPO has the \emph{subterm property},
i.e.,  $C[t] >_\mpo^= t$ for all terms $t$ and contexts $C$.

\begin{example}
\label{ex:add1}
Consider the TRS $\cR$ for addition and multiplication:
\begin{align*}
\m{add}(\m{0},y) &\to y
&
\m{mult}(\m{0},y) &\to \m{0}
\\
\m{add}(\m{s}(x),y) &\to \m{s}(\m{add}(x,y))
&
\m{mult}(\m{s}(x),y) &\to \m{add}(\m{mult}(x,y), y)
\end{align*}
By taking a precedence $\succsim$ with $\m{mult} \succ \m{add} \succ
\m{s}$ the relation $\ell >_\mpo r$ holds for all $\ell \to r \in \cR$.
For instance, $\m{add}(\m{s}(x),y) >_\mpo \m{s}(\m{add}(x,y))$ is verified
by the derivation on the left of \cref{fig:add}.  Hence, the termination of $\cR$
is concluded by \cref{prop:compatibility} together with \cref{thm:mpo}.
\end{example}

\begin{figure}[t]
\centering
\begin{prooftree}[center=false]
 \hypo{\m{add} \succ \m{s}}
    \hypo{x >_\mpo^= x}
   \infer{1}[\MPO{1}]{\m{s}(x) >_\mpo x}
  \infer{1}[$\mul$]{\{\m{s}(x),y\} >_\mpo \{x,y\}}
 \infer{1}[\MPO{3}]{\m{add}(\m{s}(x),y) >_\mpo \m{add}(x,y)}
\infer{2}[\MPO{2}]{\m{add}(\m{s}(x),y) >_\mpo \m{s}(\m{add}(x,y))}
\end{prooftree}
\hfil
\begin{prooftree}[center=false]
 \hypo{\m{add} \succ \m{s}}
    \hypo{x \gtB^= x}
   \infer{1}[\AMPO{1}]{\vphantom{|}\m{s}~x \gtB x}
  \infer{1}[$\mul$]{\{\m{s}~x,y\} \gtB \{x,y\}}
 \infer{1}[\AMPO{3}]{\m{add}~(\m{s}~x)~y \gtB \m{add}~x~y}
\infer{2}[\AMPO{2}]{\m{add}~(\m{s}~x)~y \gtB \m{s}~(\m{add}~x~y)}
\end{prooftree}
\caption{Derivations of MPO (left) and AMPO (right)}
\label{fig:add}
\end{figure}

An \emph{applicative signature} consists of constants and a single binary
symbol $\app$ called application.  Terms over an applicative signature are
called \emph{applicative terms}.  Using juxtaposition as in combinatory logic,
$\app(\ldots\app(t_1, t_2)\ldots,t_n)$ is abbreviated  to $t_1\,t_2 \cdots t_n$.  For
example, $\app(\app(\m{cons},\app(f,x)),\app(\app(\m{map},f),\mi{xs}))$ is
written as $\m{cons}~(f~x)~(\m{map}~f~\mi{xs})$.
\emph{An applicative TRS} (ATRS) is a TRS over an applicative signature,
{whereas TRSs/terms over ordinary signatures are sometimes called 
\emph{functional} TRSs/terms.}
The TRS for $\m{map}$ in the introduction is an ATRS over 
$\{ \m{nil}, \m{cons}, \m{map}, \app \}$.  Currying~\cite{KKSdeF98}
recursively transforms a term $t$ over a (possibly non-applicative)
signature $\cF$ into an applicative term $t^\circ$ as follows: $t^\circ = t$ if
$t \in \cV$, and $f\,\aseq{t^\circ}$ if $t = f(\seq{t})$. 
(Strictly speaking, the signature after currying is obtained 
by regarding each $f \in \cF$ as a constant and adding a fresh binary symbol $\circ$.)

\begin{example}
\label{ex:add2}
The curried version of the addition TRS (\cref{ex:add1}) is as follows:
\begin{align*}
\m{add}~\m{0}~y &\to y
&
\m{mult}~\m{0}~y &\to \m{0}
\\
\m{add}~(\m{s}~x)~y &\to \m{s}~(\m{add}~x~y)
&
\m{mult}~(\m{s}~x)~y &\to \m{add}~(\m{mult}~x~y)~y
\end{align*}
As  mentioned in the introduction,
the sufficient condition $\cR \subseteq {>}$ for termination cannot be
satisfied by any reduction order $>$ having the subterm property, including
MPO. If
so, for the last rule $\ell \to r$ and $\sigma = \{ y \mapsto \m{s}~x \}$
we have $\ell\sigma > r\sigma$, which leads to a contradiction:
\[
\ell\sigma
>
r\sigma
=
\underline{\m{add}}~(\m{mult}~x~(\m{s}~x))~(\m{s}~x)
>
\m{mult}~\underline{x}~(\m{s}~x)~(\m{s}~x)
>
\m{mult}~(\m{s}~x)~(\m{s}~x)
=
\ell\sigma
\]
Here, the decreases with respect to $>$ by chopping the underlined symbols
are guaranteed by the subterm property and the closure under contexts of $>$.
\end{example}

\section{Base Order}
\label{sec:base}

Toward the goal of devising a reduction order for applicative terms,
we begin with the definition of a base order which forms a well-founded order.

From now on, we assume that $\cF$ is an applicative signature,
and that it contains a special fresh constant $\APP$.
We write $\Sigma$ for the set of all constants in $\cF$.  Observe that
an arbitrary applicative term $t$ can be written in the form $f\,\aseq{t}$
with $f \in \Sigma \cup \cV$.  When $n > 0$, the term $t$ is said to be
\emph{reifiable} and the term $\APP\,(f\,\aseq[n-1]{t})\,t_n$ is denoted by
$\reify{t}$.  The term $\reify{t}$ is called the \emph{reification} of $t$.
\emph{Arity assignments} are functions from $\Sigma$ to $\bN$ with
$\alpha(\APP) = 2$.  Let $\alpha$ be an arity assignment.
A term $t = f\,\aseq{t}$ is
\emph{full} if $f \in \Sigma$ and $\alpha(f) = n$, and \emph{partial} if
$f \in \Sigma$ and $\alpha(f) > n$.
Note that a term can be neither full nor partial.

\begin{example}
\label{ex:full-partial-reify}
Let $\alpha$ be the arity assignment with
$\alpha(\APP) = \alpha(\m{map}) = \alpha(\m{cons}) = 2$,
$\alpha(\m{id}) = 1$, and $\alpha(\m{nil}) = 0$.
The term $\m{map}~f~\m{nil}$ is full, $\m{map}~\m{id}$ is 
partial, and both $\m{map}~f~x~y$ and $f\,x$ are neither full nor
partial.
These four terms are reifiable, while $x$ and $\m{map}$ are not.
The reification of $f\,x$ is $\APP\,f\,x$.
\end{example}

\begin{definition}
\label{def:ampo}
Let $\alpha$ be an arity assignment and $\succsim$ a precedence on
$\Sigma$.  The relation $\gtB$ is inductively defined on terms over
$\Sigma$ as follows: $s \gtB t$ if $s = f\, \aseq[m]{s}$ and one of the
following conditions holds.
\begin{enumerate}[~\bfseries1.]
\item
\label{ampo1}
$s_i \gtB^= t$ for some $i$.
\item
\label{ampo2}
$s$ and $t = g\,\aseq{t}$ are full, $f \succ g$, and
$s \gtB t_j$ for all $j$.
\item
\label{ampo3}
$s$ and $t = g\,\aseq{t}$ are full, 
$f \succsim g$, and $\{\seq[m]{s}\} \gtB^\mul \{\seq{t}\}$.
\item
\label{ampo4}
$s$ is full, $t$ is reifiable and not full, and $s \gtB \reify{t}$.
\item
\label{ampo5}
$s$ is reifiable and not full, and $\reify{s} \gtB t$.
\end{enumerate}
\end{definition}

The first three cases~\AMPO{1}, \AMPO{2}, and \AMPO{3} are from
the original multiset path order.

\begin{example}[continued from \cref{ex:add2}]
\label{ex:add3}
Let $\alpha$ be the arity assignment such that $\alpha(\m{add}) = 2$,
$\alpha(\m{s}) = 1$, and $\alpha(\m{0}) = 0$, and let $\succsim$ be the
same precedence as before.  Then $\ell \gtB r$ holds for all
$\ell \to r \in \cR$.  For example,
$\m{add}~(\m{s}~x)~y \gtB \m{s}~(\m{add}~x~y)$ is verified by the second
derivation in \cref{fig:add}.  Remark that the subterm property does not
hold for $\gtB$.  In fact, $\m{s}~x \gtB \m{s}$ does not hold.
\end{example}

The last two cases~\AMPO{4} and \AMPO{5} in \cref{def:ampo} are used
for resolving non-fullness.

\begin{example}
\label{ex:ampo}
Let $\alpha$ be the arity assignment of \cref{ex:full-partial-reify},
and let $\succsim$ be the precedence with $\APP \sim \m{map} \succ \m{cons} \succ \m{nil}$.
For the terms $s = \m{map}~f~(\m{cons}~x~\mi{xs})$ and 
$s' = \m{cons}~x~\mi{xs}$ we have the following derivation:
\[
\begin{prooftree}
 \hypo{\makebox[2em][l]{$\m{map} \succ \m{cons}$}}
   \hypo{\m{map} \succsim \APP}
     \hypo{x \gtB^= x}
    \infer{1}[\AMPO{1}]{s' \gtB x}
   \infer{1}[$\mul$]{\{f, s'\} \gtB^\mul \{f,x\}}
  \infer{2}[\AMPO{3}]{s \gtB \APP~f~x}
 \infer{1}[\AMPO{4}]{s \gtB f~x}
  \hypo{\m{map} \succsim \m{map}}
    \hypo{\mi{xs} \gtB^= \mi{xs}}
   \infer{1}[\AMPO{1}]{s' \gtB \mi{xs}}
  \infer{1}[$\mul$]{\{f,s'\} \gtB^\mul \{f,\mi{xs}\}}
 \infer{2}[\AMPO{3}]{s \gtB \m{map}~f~\mi{xs}}
\infer[separation=0em]{3}[\AMPO{3}]{s \gtB \m{cons}~(f~x)~(\m{map}~f~\mi{xs})}
\end{prooftree}
\]
The next two derivations illustrate case~\AMPO{5}:
\begin{align*}
&
\begin{prooftree}[center=false]
  \hypo{\APP \succsim \APP}
   \hypo{\{ \m{map}, \m{id}~f \} \gtB^\mul \{ \m{map}, f \} }
 \infer{2}[\AMPO{3}]{\APP~\m{map}~(\m{id}~f) \gtB \APP~\m{map}~f}
 \infer{1}[\AMPO{4}]{\APP~\m{map}~(\m{id}~f) \gtB \m{map}~f}
\infer{1}[\AMPO{5}]{\m{map}~(\m{id}~f) \gtB \m{map}~f}
\end{prooftree}
&&
\begin{prooftree}[center=false]
    \hypo{\APP \succ \m{id}}
      \hypo{x \gtB^= x}
    \infer{1}[\AMPO{1}]{\APP \, x \, y \gtB x}
  \infer{2}[\AMPO{2}]{\APP \, x \, y \gtB \m{id} \, x}
\infer{1}[\AMPO{5}]{x \, y \gtB \m{id} \, x}
\end{prooftree}
\end{align*}
While the former shows usefulness of reification in handling partial
application, the latter pathological derivation in fact suggests that $\gtB$
is not closed under substitutions, cf.~\cref{ex:non-stable}.
\end{example}

We prove that $\gtB$ is a well-founded order.
To this end we introduce two measures for induction.
Let $\alpha$ be an arity assignment and $\succsim$ a precedence.
Given a term $t = f\,\aseq{t}$ with $f \in \Sigma \cup \cV$,
we define $\|t\|_1$ and $\|t\|_2$ as follows:
\begin{align*}
\|t\|_1 &=
\begin{cases}
\|t_1\|_1 + \|t_2\|_1
& \text{if $f = \APP$ and $n = 2$}
\\
\|t_1\|_1 + \cdots + \|t_n\|_1 + 1
& \text{otherwise}
\end{cases}
\\
\|t\|_2 &= 
\begin{cases}
\|t_1\|_2 + \|t_2\|_2 
& \text{if $f = \APP$ and $n = 2$}
\\
\|t_1\|_2 + \cdots + \|t_n\|_2 + n
& \text{otherwise}
\end{cases}
\end{align*}
The first function counts the number of symbols except $\APP$ occurring in
the form $\APP \, t_1 \, t_2$, while the second function
computes the sum of the numbers of arguments not headed by $\APP$.

\begin{example}
For $t = \m{map}~f~\mi{xs}$ and 
$\reify{t} = \APP~(\m{map}~f)~\mi{xs}$ we have:
\begin{align*}
\|t\|_1 &= \|f\|_1 + \|\mi{xs}\|_1 + 1 
= 3
&
\|\reify{t}\|_1 &= \|\m{map}~f\|_1 + \|\mi{xs}\|_1 
= 3
\\
\|t\|_2 &= \|f\|_2 + \|\mi{xs}\|_2 + 2 
= 2
&
\|\reify{t}\|_2 &= \|\m{map}~f\|_2 + \|\mi{xs}\|_2 + 0 
= 1
\end{align*}
\end{example}

\begin{lemma}
\label{lem:measures}
Let $t = f\,\aseq{t}$ with $f \in \Sigma \cup \cV$.
\begin{enumerate}[(1)]
\item
$\|t\|_1 > 0$.
\item
$\|t\|_1 > \|t_i\|_1$ for all $1 \leqslant i \leqslant n$.
\item 
$\|t\|_1 \geqslant \|\reify{t}\|_1$ and $\|t\|_2 > \|\reify{t}\|_2$ if 
$t$ is reifiable but not full.
\end{enumerate}
\end{lemma}
\begin{proof}
Easy induction on $t$ proves the first claim, from which the second
claim follows.  To prove the third claim, consider a term 
$t = f\, \aseq{t}$
with $f \in \Sigma \cup \cV$ that is reifiable and not full.  From the
conditions $n > 0$ and $(f,n) \neq (\APP,2)$ follow.
Therefore, the desired inequalities
\begin{align*}
\|t\|_1
&= (\|t_1\|_1 + \cdots + \|t_{n-1}\|_1 + 1) + \|t_n\|_1
\geqslant \|f\,\aseq[n-1]{t}\|_1 + \|t_n\|_1
= \|\reify{t}\|_1
\\
\|t\|_2
&= \|t_1\|_2 + \cdots + \|t_n\|_2 + n
> \|t_1\|_2 + \cdots + \|t_{n-1}\|_2 + n-1 + \|t_n\|_2
= \|\reify{t}\|_2
\end{align*}
hold.
\qed
\end{proof}

We note that in the case of (3) $\|t\|_1 = \|\hat{t}\|_1$ does not always hold.
This can be exemplified by $t = \APP\, \APP \, \APP \,\APP$
for which $\|t\|_1 = 4 > 3 = \|\hat{t}\|_1$ holds.

Although \AMPO{4} and \AMPO{5} increase sizes of terms by reification, the
measures serve as decreasing measures for induction.

\begin{lemma}
\label{lem:transitivity}
The relation $\gtB$ is transitive.
\end{lemma}
\begin{proof}
Suppose $s \gtB t \gtB u$, where $s = f\,\aseq[k]{s}$,
$t = g\,\aseq[m]{t}$, and $f, g \in \Sigma \cup \cV$.  We show 
$s \gtB u$ by induction on
$(\|s\|_1 + \|t\|_1 + \|u\|_1, \|s\|_2 + \|t\|_2 + \|u\|_2)$ with respect
to lexicographic comparison of pairs of natural numbers.  We distinguish
several cases depending on how $s \gtB t$ and $t \gtB u$
are derived, where we apply the induction hypothesis with decreases in the
first component using \cref{lem:measures}(2) unless otherwise stated.
\begin{enumerate}[(i)]
\item
Suppose $s \gtB t$ is derived by \AMPO{1}. Then
$s_i \gtB t$ holds for some $i$.  As we have
$s_i \gtB t \gtB u$, the induction hypothesis yields 
$s_i \gtB u$.  Thus, $s \gtB u$ follows by
\AMPO{1}.
\item
Suppose $s \gtB t$ is derived by \AMPO{2} or \AMPO{3} and $t \gtB u$ is
derived by \AMPO{1}.  First, observe that $s \gtB t_j$ holds for all $j$,
even when $s \gtB t$ is derived by \AMPO{3}: from 
$\{ \seq[k]{s} \} \gtB^\mul \{ \seq[m]{t} \}$ we obtain $s_i \gtB^= t_j$ for some
$i$, so $s \gtB t_j$ follows by \AMPO{1}.  Now we prove $s \gtB u$.
From $s \gtB u$ we have $t_j \gtB^= u$ for some $j$.  If $t_j = u$ then 
$s \gtB u$ holds.  Otherwise, $s \gtB t_j \gtB u$, from which $s \gtB u$
follows by the induction hypothesis.
\item
Suppose $s \gtB t$ and $t \gtB u$ are derived by \AMPO{2} or \AMPO{3}.
Then $u$ is of the form $h\aseq{u}$ with $t \gtB u_j$ for all $j$.  
\begin{itemize}
\item
If both are due to \AMPO{3}, then the relations
\begin{align*}
f &\succsim g \succsim h
&
\{\seq[k]{s}\} &\gtB^\mul \{\seq[m]{t}\} \gtB^\mul \{\seq[n]{u}\}
\end{align*}
hold.  As $\{\seq[k]{s}\} \gtB^\mul \{\seq[n]{u}\}$ follows by the induction
hypothesis, we obtain $s \gtB u$ by \AMPO{3}.
\item
Otherwise, $f \succ h$ holds. Again, we can show
$t \gtB u_j$ for all $j$. Since we have $s \gtB t \gtB u_j$,
the induction hypothesis yields $s \gtB u_j$ for all $j$.
Hence, we obtain $s \gtB u$ by \AMPO{2}.
\end{itemize}
\item
Suppose $s \gtB t$ is derived by \AMPO{2} or \AMPO{3} and $t \gtB u$ is
by \AMPO{4}.  Then $s$ is full, $u$ is reifiable and not full, and
$t \gtB \reify{u}$ holds.  Since $\|u\|_1 \geqslant \|\reify{u}\|_1$ and 
$\|u\|_2 > \|\reify{u}\|_2$ by \cref{lem:measures}(3), the induction
hypothesis with a decrease in the second component leads to 
$s \gtB \reify{u}$. Therefore, $s \gtB u$ is derived by \AMPO{4}.
\item
Suppose $s \gtB t$ is derived by \AMPO{2} or \AMPO{3} and $t \gtB u$ is
by \AMPO{5}.  Then $t$ is full and not full.  So this case never happens.
\item
Suppose $s \gtB t$ is derived by \AMPO{4}.  Then $s$ is full, $t$ is
reifiable and not full, and $s \gtB \reify{t} \gtB u$ holds.
Since $\|t\|_1 \geqslant \|\reify{t}\|_1$ and $\|t\|_2 > \|\reify{t}\|_2$ by
\cref{lem:measures}(3),
the induction hypothesis applies to $s \gtB \reify{t} \gtB u$
(with a decrease in the second component). Therefore, $s \gtB u$ holds.
\item
Suppose $s \gtB t$ is derived by \AMPO{5}.  Then $s$ is
reifiable and not full, and $\reify{s} \gtB t$.  
Again, with a decrease in the second component
the induction hypothesis applies to $\reify{s} \gtB t \gtB u$.
Thus, we have $\reify{s} \gtB u$ from which $s \gtB u$ follows by \AMPO{5}.
\qedhere
\end{enumerate}
\end{proof}

Well-foundedness of MPO was originally shown by the Kruskal Tree
Theorem~\cite{D82}. However, it relies on the subterm property of the order.
Thus, we employ the alternative proof method by Buchholz~\cite{B95}.

Recall the notion of lexicographic product. Let $A$ be a set.  Consider $n$
quasi-orders $\seq{\gtrsim}$
on sets $\seq{A}$, respectively.
Their \emph{lexicographic product} is the relation $\gg$ on $A_1 \times
\cdots \times A_n$ defined as follows: $(\seq{a}) \gg (\seq{b})$ if there
exists an index $1 \leqslant k \leqslant n$ such that $a_k >_k b_k$ and 
$a_i \gtrsim_i b_i$
for all $i < k$.  It is known that $\gg$ is a well-founded order if
every $>_i$ is a well-founded order, where $>_i$ is the strict part
${\gtrsim_i} \setminus {\lesssim_i}$ of $\gtrsim_i$.

Below, we assume well-foundedness of $\succ$.
Let $\SN$ denote the set of all terms
that have no infinite descending sequence with respect to $\gtB$.
We define the set $U$ of tuples as 
$\Sigma \times \cM(\SN) \times \bN \times \bN$
and the relation $\gg$ on $U$ as the lexicographic product of 
$\succsim$, the reflexive closure of $\gtB^\mul$,
and two copies of the ordinary order on $\bN$.
Then $\gg$ is a well-founded order on $U$.  For a full term $s = f\aseq{s}$
with $f \in \Sigma$ and a term $t$ we define $\phi(s,t)$ as
\(
\phi(s,t) = (f,\{\seq{s}\},\|t\|_1, \|t\|_2)
\).
Note that $\phi(s,t)$ belongs to $U$ whenever $\seq{s} \in \SN$.

\begin{lemma}
\label{lem:sn-key}
Let $s = f\aseq[m]{s}$ be a full term with $f \in \Sigma$.  If 
$\seq[m]{s} \in \SN$ and $s \gtB t$ then $t \in \SN$.
\end{lemma}
\begin{proof}
Suppose $\seq[m]{s} \in \SN$ and $s \gtB t$.
We show 
$t \in \SN$ by well-founded induction on $\phi(s,t)$.  Depending on how 
$s \gtB t$ is derived, we distinguish cases.
\begin{enumerate}[1.]
\item
If $s \gtB t$ is derived by \AMPO{1} then $s_i = t$ or $s_i \gtB t$ holds.
In either case, the claim follows from $s_i \in \SN$.
\item
If $s \gtB t$ is derived by \AMPO{2} then $t$ is a full term 
$g\,\aseq{t}$ with $g \in \Sigma$ and $f \succ g$, and
$s \gtB t_j$ for all $j$.  Since $\phi(s, t) \gg \phi(s, t_j)$ with a
decrease at the third, $t_j \in \SN$ for each $j$.  In order to show 
$t \in \SN$ it is enough to show that $u \in \SN$ whenever $t \gtB u$.  So
suppose $t \gtB u$.
We have $\phi(s,t) \gg \phi(t,u)$ as $f \succ g$.
Hence, the claim follows by the induction hypothesis.  
\item
If $s \gtB t$ is derived by \AMPO{3} then we apply an argument similar to
the previous case, using a decrease in the second component of the
induction measure $\phi(s, t)$.
\item
If $s \gtB t$ is derived by \AMPO{4} then $t$ is reifiable and not full, and
$s \gtB \reify{t}$.  Note that $t$ is of the form $g\,\aseq{t}$
with $g \in \Sigma \cup \cV$ and $n > 0$. In order to show $t \in \SN$ it
is enough to show that $u \in \SN$ whenever $t \gtB u$.  So suppose
$t \gtB u$.  Using \cref{lem:measures}(3) we have
$\phi(s,t) \gg \phi(s,\reify{t})$ as the fourth component decreases.  
So, $\reify{t} \in \SN$ by the induction hypothesis.  Thus,
$u \in \SN$ is concluded if we show $\reify{t} \gtB^* u$.
Since $t$ is not full, $t \gtB u$ is derived by \AMPO{1} or \AMPO{5}.
In the latter case, $\reify{t} \gtB u$ is immediate.  In the former case
$t_j \gtB^= u$ holds for some $j$.  If $j < n$ then
$\reify{t} \gtB g\, \aseq[n-1]{t} \gtB t_j \gtB^= u$.
If $j = n$ then $\reify{t} \gtB t_j \gtB^= u$. 
So $\reify{t} \gtB^* u$ holds.
\item
If $s \gtB t$ is derived by \AMPO{5} then $s$ is not full, which contradicts 
the assumption.
\qedhere
\end{enumerate}
\end{proof}

\begin{lemma}
\label{lem:well-foundedness}
The relation $\gtB$ is well-founded.
\end{lemma}
\begin{proof}
Let $s = f\aseq{s}$ with $f \in \Sigma \cup \cV$.  We show $s \in \SN$ by
induction on $|s|$.  By the induction hypothesis we have
$\seq[n]{s} \in \SN$.  It follows if we prove $t \in \SN$ for all terms $t$
with $s \gtB t$.  Let $t$ be a term with $s \gtB t$.   
If $s$ is full, then we conclude by \cref{lem:sn-key}.
Suppose that $s$ is not full.  Then $s \gtB t$ is due to \AMPO{1}
or \AMPO{5}. In the former case we have $s_i \gtB^= t$. Since 
$s_i \in \SN$, we obtain $t \in \SN$.
In the latter case  $s$ is reifiable and not full
and $\reify{s} = \APP \, (f\,\aseq[n-1]{s}) \, s_n \gtB t$
holds.
Because $f\,\aseq[n-1]{s} \in \SN$ by the induction hypothesis, 
$t \in \SN$ follows by \cref{lem:sn-key}.
\qed
\end{proof}

By combining the lemma with \cref{lem:transitivity} the announced claim
is concluded.

\begin{theorem}
\label{thm:wfo}
The relation $\gtB$ is a well-founded order if $\succ$ is well-founded.
\end{theorem}

\section{Eta-Restriction}
\label{sec:mono}

The order $\gtB$ is not closed under contexts.

\begin{example}
\label{ex:monotonicity}
Let $\alpha$ be the arity assignment with
$\alpha(\APP) = \alpha(\m{g}) = 2$ and $\alpha(\m{f}) = 1$.
\begin{enumerate}[(a)]
\item
Assume $\m{f} \sim \m{g} \succ \APP$.
As $\m{f}\,\m{g} \gtB \APP\,\m{g}\,\m{g}$ holds,
$\m{f}\,\m{g} \gtB \m{g}\,\m{g}$ follows by \AMPO{4}.
However,
$\m{f}\,\m{g}\,x \gtB \m{g}\,\m{g}\,x$ does not hold
as $\m{g}~\m{g}~x \gtB \m{f}\,\m{g}\,x$ follows by \AMPO{4}.
\item
Assume $\m{f} \succ \m{g}$. Then $\m{g}\,\m{f} \gtB \m{f}$ holds by
\AMPO{1}.  However, $\m{g}\,\m{f}\,x \gtB \m{f}\,x$ does not hold.  In
fact, \AMPO{1} does not apply because neither $\m{f} \gtB^= \m{f}\,x$ nor
$x \gtB^= \m{f}\,x$ holds.
\end{enumerate}
Note that the right-hand side of $\m{f}\,\m{g} \gtB \m{g}\,\m{g}$ is
partial, while the left-hand side of $\m{g}\,\m{f} \gtB \m{f}$ is partial.
\end{example}

The example (a) is avoided by imposing maximality of $\APP$
on the precedence. We address  (b) by a technique
akin to $\eta$-saturation~\cite{HMZ13}.

\begin{definition}
\label{def:eta}
The $\eta$-restriction $\gtE$ of $\gtB$ is recursively defined as
follows: $s \gtE t$ if $s \gtB t$ holds, and when $s$ is partial 
$s\,x \gtE t\,x$ additionally holds.  Here $x$ is a fresh variable not in
$\Var(s)$.
\end{definition}

If $s = f\,\aseq[n]{s}$ with $f \in \Sigma$ is partial,
by definition $s \gtE t$ if and only if $s \gtB t$ and
\(
s \, x_{n+1} \cdots x_k \gtB
t \, x_{n+1} \cdots x_k 
\)
for all $n + 1 \leqslant k \leqslant \alpha(f)$,
where $x_{n+1},\ldots,x_{\alpha(f)}$ are distinct fresh variables not occurring in $\Var(s)$.
We can see that
\cref{ex:monotonicity}(b) does not go for $\gtE$:
The inequality $\m{g}\,\m{f} \gtE \m{f}$ does not hold because $\m{g}\,\m{f}\,x \gtB \m{f}\,x$ does not.

The rest of the section is for proving that $\gtE$ is 
closed under contexts and is a
well-founded order if $\APP$ is maximum in the precedence.
First, we prove a key lemma using the assumption.

\begin{lemma}
\label{lem:reification}
Suppose $\APP$ is maximum in the precedence.
If a term of the form $s\,t$ is full then $\APP\,s\,t \gtB s\,t$.
\end{lemma}
\begin{proof}
Suppose $s\,t$ is full, where $s = f\,\aseq[n]{s}$ with $f \in \Sigma$.  We
show that $\APP\,s\,t \gtB s\,t$ is derived by \AMPO{3}.  As $\APP$ is
maximum, $\APP \succsim f$ holds. It remains to show
$\{s,t\} \gtB^\mul \{\seq[n]{s},t\}$.
If $n = 0$ then $\{s,t\} = \{f,t\} \gtB^\mul \{t\}$.
If $n > 0$ then \AMPO{1} derives
$s = f\,\aseq{s} \gtB s_i$ for all $i$.  
Thus, $\{s,t\} \gtB^\mul \{\seq[n]{s},t\}$.
\qed
\end{proof}

When $s$ is partial, $s \gtE t$ demands $s \gtB t$ and $s\,x \gtE t\,x$
for some variable $x$ not in $\Var(s)$.  Due to the next property, 
$x \notin \Var(s)$ is equivalent to $x \notin \Var(s) \cup \Var(t)$.

\begin{lemma}
\label{lem:sensibility}
If $s \gtB t$ then $\Var(s) \supseteq  \Var(t)$.
\end{lemma}

We show that $s\,x \gtB t\,x$ with a fresh variable $x$ is equivalent
to $s\,u \gtB t\,u$ for all terms $u$.  Although $\gtB$ is \emph{not}
closed under substitutions (see \cref{ex:ampo,ex:non-stable}),
it \emph{is} closed under head-preserving substitutions.
For a term $t$ we define
\[
\HV(t) = \{ x \in \cV \mid \text{$t = C[x\,u]$ for some context $C$ and term $u$} \}.
\]
Variables in $\HV(t)$ are called \emph{head variables}~\cite{AY04}.
We say that a substitution $\sigma$ is \emph{head-preserving} for $t$ if
$\sigma(x) = x$ for all $x \in \HV(t)$. 
For example, $\HV(t) = \{ f \}$ for
$t = \m{cons}~(f~x)~(\m{map}~f~\mi{xs})$. So $\{ \mi{xs} \mapsto \m{nil}
\}$ is head-preserving for $t$, but $\{ f \mapsto \m{id} \}$ is \emph{not}.

\begin{lemma}
\label{lem:head-preservation}
Suppose $\sigma$ is head-preserving for $s$ and $t$.
If $s \gtB t$ then $s\sigma \gtB t\sigma$.
\end{lemma}
\begin{proof}
The claim is shown by easy induction on the derivation of $s \gtB t$.
\qed
\end{proof}

For a term $s = f\,\aseq{s}$ with $f \in \Sigma \cup \cV$, we define
$\|s\|_\eta$ as $\max\{0,\alpha(f) - n\}$ if $f \in \Sigma$, and $0$
otherwise. If $s$ is partial then $\|s\|_\eta > \|s\,x\|_\eta$.  This
measure $\|s\|_\eta$ indicates the number of arguments to make $s$ 
non-partial.

\begin{lemma}
\label{lem:head-preservation-eta}
Suppose $\sigma$ is head-preserving for $s$ and $t$.  If $s \gtE t$ then
$s\sigma \gtE t\sigma$.
\end{lemma}
\begin{proof}
Suppose $s = f\,\aseq{s} \gtE t$ with $f \in \Sigma \cup \cV$.
We show $s\sigma \gtE t\sigma$ by induction on
$\|s\|_\eta$.
By \cref{lem:head-preservation} we have $s\sigma \gtB t\sigma$.  
If $s\sigma$ is not partial then we immediately conclude $s\sigma \gtE t\sigma$.
Otherwise, we have to show $(s\sigma)\,x \gtE (t\sigma)\,x$ for a variable $x$
not in $\Var(s\sigma)$.
Since $s\sigma$ is partial and $\sigma$ is head preserving for $s$,
the term $s$ must be partial too.
So from $s \gtE t$ we have $s\,y \gtE t\,y$ for some variable $y \notin \Var(s)$.
By \cref{lem:sensibility} we have $y \notin \Var(t)$.
Define the substitution $\tau$ as follows: $\tau(y) = x$ and
$\tau(z) = \sigma(z)$ for other variables $z$.
Clearly, $\tau$ is head-preserving for $s\,y$ and $t\,y$,
and moreover, since $y$ occurs in neither $s$ nor $t$,
we obtain $s\tau = s\sigma$ and $t\tau = t\sigma$.
Finally, as $\|s\|_\eta > \|s\,y\|_\eta$,
the induction hypothesis for $s\,y \gtE t\,y$
yields 
\(
(s\sigma)\,x = (s\,y)\tau \gtE (t\,y)\tau = (t\sigma)\,x
\).
Hence, $s\sigma \gtE t\sigma$ is concluded.
\qed
\end{proof}

The technical lemma below allows us to conclude $s\,u \gtB t\,u$
(without $\eta$) from $s \gtE t$ (with $\eta$).

\begin{lemma}
\label{lem:right}
Suppose $\APP$ is maximum in the precedence.
Let $u$ be a term.  If $s \gtE t$ then $s\,u \gtB t\,u$.
\end{lemma}
\begin{proof}
Suppose $s \gtE t$ and $s = f\,\aseq{s}$ with $f \in \Sigma \cup \cV$.  
\begin{itemize}
\item
If $s$ is partial then $s \gtE t$ yields 
$s\,x \gtB t\,x$ for some variable $x$ not in $\Var(s)$.  According
to \cref{lem:sensibility}, the variable $x$ is neither in $\Var(t)$. So the
substitution $\sigma = \{ x \mapsto u \}$ is head-preserving for $s\,x$ and
$t\,x$.  Therefore, $s\,u = (s\,x)\sigma \gtB (t\,x)\sigma = t\,u$
is obtained by \cref{lem:head-preservation}.
\item
If $s$ is not partial then $s\,u$ is reifiable and not full.  By \AMPO{1}
we can derive $\APP\,s\,u \gtB t$ from $s \gtB t$ and
$\APP\,s\,u \gtB u$.  Moreover, $\{s,u\} \gtB^\mul \{t,u\}$
holds.  Thus, $\APP\,s\,u \gtB \APP\,t\,u$ follows by \AMPO{3}.
If $t\,u$ is not full,  the claim follows by \AMPO{4} and \AMPO{5}.
Otherwise, according to \cref{lem:reification},
$\APP\,t\,u \gtB t\,u$ holds,
so $\APP\,s\,u \gtB t\,u$ follows by transitivity
(\cref{lem:transitivity}).  Therefore, $s\,u \gtB t\,u$ is derived
by \AMPO{5}.
\qedhere
\end{itemize}
\end{proof}

\begin{lemma}
\label{lem:left}
Let $u$ be a term.  If $s \gtB t$ then $u\,s \gtB u\,t$.
\end{lemma}
\begin{proof}
Suppose $s \gtB t$ and $u = f\,\aseq{u}$ with $f \in \Sigma \cup \cV$.  
\begin{itemize}
\item
If $u\,s$ is full then by \AMPO{1} we can derive $u\,s \gtB t$ and
$u\,s \gtB u_i$ for all $i$.  Moreover,
$\{\seq{u},s\} \gtB^\mul \{\seq{u},t\}$ holds.  Thus,
$u\,s \gtB u\,t$ follows by \AMPO{3}.
\item
Otherwise, $u\,s$ and $u\,t$ are reifiable and not full.  By \AMPO{1} we can
derive $\APP\,u\,s \gtB t$ from $s \gtB t$ and
$\APP\,u\,s \gtB u$.  Moreover, $\{u,s\} \gtB^\mul \{u,t\}$
holds.  Thus, $\APP\,u\,s \gtB \APP\,u\,t$ follows by \AMPO{3}.
Therefore, $u\,s \gtB u\,t$ is obtained by \AMPO{4} and \AMPO{5}.
\qedhere
\end{itemize}
\end{proof}

As a consequence of \cref{lem:left}, if $s \gtE t$ (with $\eta$) then
$u\,s \gtB u\,t$.

\begin{lemma}
\label{lem:contexts}
Suppose $\APP$ is maximum in the precedence.
Let $C$ be a context.  If $s \gtE t$ then $C[s] \gtB C[t]$.
\end{lemma}
\begin{proof}
Easy induction on $C$ with using \cref{lem:left,lem:right}.
\end{proof}

Now we are ready to show that $\gtE$ is closed under contexts.

\begin{theorem}
\label{thm:contexts}
Suppose $\APP$ is maximum in the precedence and $C$ is a context.  
If $s \gtE t$ then $C[s] \gtE C[t]$.
\end{theorem}
\begin{proof}
Suppose $s \gtE t$.  We show $C[s] \gtE C[t]$
by induction on $\|C[s]\|_\eta$.  Suppose $s \gtB t$.  By
\cref{lem:contexts} we have $C[s] \gtB C[t]$.  When $C[s]$ is
partial, we have to show $C'[s] \gtB C'[t]$ for 
$C' = C\,x$,
where $x$ is a variable not in $\Var(C[s])$.  As
$\|C[s]\|_\eta > \|C'[s]\|_\eta$, the induction hypothesis applies.
\end{proof}

Using \cref{thm:contexts}, we can prove that the transitivity of $\gtB$
(\cref{lem:transitivity}) carries over to $\gtE$.

\begin{lemma}
\label{lem:eta-transitivity}
The relation $\gtE$ is transitive, provided that $\APP$ is maximum in
the precedence.
\end{lemma}
\begin{proof}
Suppose $s \gtE t \gtE u$.  We prove $s \gtE u$ by induction on
$\|s\|_\eta$. The transitivity of $\gtB$ (\cref{lem:transitivity}) yields
$s \gtB u$.  When $s$ is partial, we also have to show $s\,x \gtE u\,x$ for some
$x \in \cV \setminus \Var(s)$.  By \cref{thm:contexts} the relations
$s\,x \gtE t\,x \gtE u\,x$ follow from the assumption. As we have
$\|s\|_\eta > \|s\,x\|_\eta$, the induction hypothesis yields
$s\,x \gtE u\,x$.
\end{proof}

The well-foundedness of
$\gtB$ (\cref{lem:well-foundedness}) also carries over to $\gtE$ as $\gtE$
is a subrelation of $\gtB$.  

\begin{lemma}
\label{lem:eta-wfo}
The relation $\gtE$ is a well-founded order if $\succ$ is
well-founded.
\end{lemma}

Hence, the following corollary is obtained.

\begin{corollary}
\label{cor:mono-wfo}
The relation $\gtE$ is a well-founded order that is closed under contexts,
if $\APP$ is maximum in the precedence and $\succ$ is well-founded.
\end{corollary}

\section{Reduction Order}
\label{sec:stability}

The relation $\gtE$ is \emph{not} closed under substitutions.

\begin{example}[continued from \cref{ex:ampo}]
\label{ex:non-stable}
Although $x\,y \gtB \m{id}\,y$ holds,
$\m{id}\,y = (x\,y)\sigma \gtB (\m{id}\,y)\sigma = \m{id}\,y$ does not hold for
$\sigma = \{ x \mapsto \m{id} \}$.  
As $x\,y$ is not partial, we conclude that $\gtE$ is not closed under
substitutions and therefore not a reduction order.
\end{example}

We overcome the problem using an existing technique to construct a
reduction order from a 
well-founded order $>$ on
\emph{ground} terms that is closed under ground contexts.
We call a substitution $\sigma$ \emph{grounding} for a set $T$ of
terms if $t\sigma$ is ground for all $t \in T$.

\begin{definition}
\label{def:gtM}
Let $\alpha$ be an arity assignment and $\succsim$ a precedence.  The
relation $\gtM$ on terms is defined as follows: $s \gtM t$
if $s\sigma \gtE t\sigma$ for all grounding substitutions $\sigma$ for $\{ s, t \}$.
\end{definition}

The relation $\gtM$ is our applicative multiset path order (AMPO).

\begin{theorem}
\label{thm:ampo}
Assume $\APP$ is maximum in the precedence.
The relation $\gtM$ is a reduction order if $\succ$ is well-founded.
\end{theorem}
\begin{proof}
According to \cite[Exercise~5.6]{BN98}, the relation defined as
in \cref{def:gtM} forms a reduction order if the underlying relation is
a well-founded order that is closed under ground contexts.
Thus, \cref{cor:mono-wfo} entails the
claim.\footnote{%
Here we also use the assumption $\APP \in \Sigma$ to ensure existence of 
a ground term.}
\end{proof}

We need a method to check if $s \gtM t$ holds for given terms $s$ and $t$.
We present a sufficient condition for $s \gtM t$ based on absence of
head variables in $s$.  A term $t$ is \emph{head variable free}
if $\HV(t) = \emptyset$.  For example, the left-hand side of the rule 
$\m{map}~f~(\m{cons}~x~\mi{xs}) \to \m{cons}~(f~x)~(\m{map}~f~\mi{xs})$
is head variable free, while the right-hand side is not because of the subterm $f~x$.

Below, we assume that $\APP$ is maximum in the precedence.

\begin{lemma}
\label{lem:hvf}
Let $\sigma$ be a substitution.
If $s \gtB t$ and $s$ is head variable free then
$s\sigma \gtB t\sigma$.
\end{lemma}
\begin{proof}
Suppose $s \gtB t$ and $s$ is head variable free.  
We show $s\sigma \gtB t\sigma$ by
induction on the derivation of $s \gtB t$.  Since $s$ is head
variable free, $s$ can be written as $f\,\aseq[m]{s}$ with $f \in \Sigma$.  
So $s\sigma = f\,(s_1\sigma)\cdots(s_m\sigma)$.
\begin{enumerate}[1.]
\item
If $s \gtB t$ is derived by \AMPO{1} then $s_i \gtB^= t$
for some $i$.  Since $s$ is head variable free, the subterm $s_i$ is head
variable free too.  By the induction hypothesis or stability of $=$ we
obtain $s_i\sigma \gtB t\sigma$.  Therefore,
$s\sigma \gtB t\sigma$ is derived by~\AMPO{1}.
\item
If $s \gtB t$ is derived by \AMPO{2} then $s\sigma \gtB t\sigma$ 
follows from the induction hypothesis and \AMPO{2}.
\item
If $s \gtB t$ is derived by \AMPO{3} then $s\sigma \gtB t\sigma$
follows from the induction hypothesis and \AMPO{3}.
\item
If $s \gtB t$ is derived by \AMPO{4} then $s$ is full, $t$ is 
reifiable and not full, and $s \gtB \reify{t}$.
If $t\sigma$ is full then the inequalities
$s\sigma \gtB \reify{t}\sigma \gtB t\sigma$
by the induction hypothesis and \cref{lem:reification}.  Thus,
$s\sigma \gtB t\sigma$ follows by transitivity
(\cref{lem:transitivity}).  Otherwise, the induction hypothesis yields
$s\sigma \gtB \reify{t}\sigma$.
Thus, $s\sigma \gtB t\sigma$ is derived by \AMPO{4}.
\item
If $s \gtB t$ is derived by \AMPO{5} then $s$ is reifiable and not full, and $\reify{s} \gtB t$.
Because $s$ is
head variable free, $\reify{s}$ is also head variable free.  So the induction
hypothesis yields $\reify{s}\sigma \gtB t\sigma$.
Therefore, $s\sigma \gtB t\sigma$ is derived by \AMPO{5}.
\qedhere
\end{enumerate}
\end{proof}

\begin{lemma}
\label{lem:hvf-eta}
Let $\sigma$ be a substitution.
If $s \gtE t$ and $s$ is head variable free then $s\sigma \gtE t\sigma$.
\end{lemma}
\begin{proof}
Suppose $s \gtE t$ and $s$ is head variable free.  By induction on
$\|s\|_\eta$ we show $s\sigma \gtE t\sigma$.
As $s \gtE t$, by definition $s \gtB t$ holds. Thus, $s\sigma \gtB t\sigma$
follows by \cref{lem:hvf}.
If $s\sigma$ is not partial then we are done.
When $s\sigma$ is partial, we have to show
$(s\sigma)\,x \gtE (t\sigma)\,x$.  Since $s$ is head variable free, $s$ is
of the form $f\,\aseq{s}$ with $f \in \Sigma$.  So
$s\sigma = f\,(s_1\sigma)\cdots(s_n\sigma)$.  As $s\sigma$ is partial, $s$
is partial too.
Thus, $sx \gtE tx$ follows from $s \gtE t$.  Here $x$ is a variable not in
$\Var(s)$, and thus $x$ is neither in $\Var(t)$ by \cref{lem:sensibility}.
Define the substitution $\tau$ as follows: $\tau(x) = x$ and
$\tau(y) = \sigma(y)$ for other variables $y$.  
Since $s\,x$ is head variable free and $\|s\,x\|_\eta > \|s\|_\eta$,
the induction hypothesis yields
$(s\sigma)\,x = (s\,x)\tau \gtE (t\,x)\tau = (t\sigma)\,x$.
Hence, $s\sigma \gtE t\sigma$ follows.
\end{proof}

\begin{theorem}
\label{thm:hvf}
Assume that $\APP$ is maximum in the precedence.
If $s \gtE t$ and $s$ is head variable free then $s \gtM t$.
\end{theorem}
\begin{proof}
Immediate from \cref{lem:hvf-eta}.
\end{proof}

We relate Dershowitz's multiset path order and ours using currying.
Consider an MPO $>_\mpo$ and functional terms $s$ and $t$.  If we set
$\alpha(f) = n$ (by a slight abuse of notation) for every $n$-ary function
$f$, the curried terms $s^\circ$ and $t^\circ$ are full and head variable
free.  For the AMPO $>_\ampo$ induced by $\alpha$ and the corresponding
precedence, we show the \emph{gracefulness} property~\cite{BBWW17,BWW17}:

\begin{theorem}
\label{thm:gracefulness}
If $s >_\mpo t$ then $s^\circ \gtM t^\circ$.
\end{theorem}
\begin{proof}
The relation $s^\circ \gtB t^\circ$ is proved by straightforward
induction on the derivation of $s >_\mpo t$.  As $s^\circ$ is full,
$s^\circ \gtE t^\circ$ holds by definition. Consider an arbitrary
grounding substitution $\sigma$ for $\{s^\circ,t^\circ\}$.  Since $s^\circ$
is head variable free, $\sigma$ is head-preserving for $s^\circ$.  So
$s^\circ\sigma \gtE t^\circ\sigma$ follows by \cref{lem:head-preservation}.
Hence, $s^\circ \gtM t^\circ$.
\qed
\end{proof}

\section{Evaluation}
\label{sec:evaluation}

We evaluate our applicative multiset path order by examples and experimental data.
An ATRS is said to be \emph{left head variable free} if the left-hand side of every
rule is head variable free. The next termination criterion is an immediate
consequence of \cref{prop:compatibility,thm:ampo,thm:hvf}
and is decidable, see \cref{thm:decidability} in \cref{sec:appendix}.

\begin{corollary}
\label{cor:termination}
Suppose that 
a precedence $\succsim$ is well-founded and has $\APP$ as maximum.  
A left head variable free ATRS $\cR$ is
terminating if $\cR \subseteq {\gtE}$ for $\gtE$ induced from an arity assignment $\alpha$ and
$\succsim$.
\end{corollary}

By examples, we demonstrate termination proving with our order.

\begin{example}
\label{ex:map}
First, consider the left head variable free ATRS
\begin{align*}
\m{append}~\m{nil}~\mi{ys} &\to \mi{ys}
&
\m{append}~(\m{cons}~x~\mi{xs})~\mi{ys}
&\to \m{cons}~x~(\m{append}~\mi{xs}~\mi{ys})
&
\m{id}~x &\to x
\\
\m{map}~f~\m{nil} &\to \m{nil}
&
\m{map}~f~(\m{cons}~x~\mi{xs})
&\to \m{cons}~(f~x)~(\m{map}~f~\mi{xs})
\\
\m{map}~\m{id} &\to \m{id}
&
\m{map}~f~(\m{append}~\mi{xs}~\mi{ys})
&\to \makebox[0mm][l]{$\m{append}~(\m{map}~f~\mi{xs})~(\m{map}~f~\mi{ys})$}
\end{align*}
which is a complete ATRS for the equational system in the introduction.
(Here, the equation of $\m{map}$ for $\m{cons}~x~\m{nil}$
is generalized to a rule for $\m{cons}~x~\mi{xs}$.)
The termination is shown by taking 
$\alpha(\APP) = \alpha(\m{cons}) = \alpha(\m{append}) = \alpha(\m{map}) = 2$, 
$\alpha(\m{id}) = 1$, and $\alpha(\m{nil}) = 0$ as well as
$\APP \sim \m{map} \succ \m{append} \succ \m{id} \succ \m{cons} \succ \m{nil}$.
Note that $\m{map}~\m{id} \gtE \m{id}$ is equivalent to
$\m{map}~\m{id} \gtB \m{id}$ and $\m{map}~\m{id}~x \gtB \m{id}~x$,
the latter of which uses $\m{map} \succ \m{id}$.
\end{example}

\begin{example}
\label{ex:filter}
Next, consider the ATRS for the filter function:
\begin{align*}
\m{filter}\,f\,\m{nil} &\to \m{nil}
&
\m{if}\,\m{true}\,x\,y &\to x
\\
\m{filter}\,f\,(\m{cons}\,x\,\mi{xs}) 
&\to \m{if}\,(f\,x)\,(\m{cons}\,x\,(\m{filter}\,f\,\mi{xs}))\,(\m{filter}\,f\,\mi{xs})
&
\m{if}\,\m{false}\,x\,y &\to y
\end{align*}
The termination is shown by taking $\alpha(\APP) = 2$, $\alpha(\m{if}) = 3$,
$\alpha(\m{cons}) = \alpha(\m{filter}) = 2$, and
$\alpha(\m{nil}) = \alpha(\m{true}) = \alpha(\m{false}) = 0$ as well as
$\APP \sim \m{filter} \succ \m{if} \succ \m{cons} \succ \m{nil}$.
\end{example}

Below is an example in which choice of arity assignment matters.

\begin{example}
\label{ex:zip-with}
Finally, consider the ATRS for the zip-with function:
\begin{align*}
\m{zip}~f~\m{nil}~\m{nil} &\to \m{nil}
\\
\m{zip}~f~(\m{cons}~x~\mi{xs})~(\m{cons}~y~\mi{ys})
&\to \m{cons}~(f~x~y)~(\m{zip}~f~\mi{xs}~\mi{ys})
\end{align*}
The termination is shown by
$\alpha(\APP) = \alpha(\m{cons}) = 2$, and
$\alpha(\m{nil}) = \alpha(\m{zip}) = 0$ as well as
$\APP \sim \m{zip} \succ \m{cons} \sim \m{nil}$.  Note that
the second rule cannot be ordered if we modify $\alpha(\m{zip}) = 0$ to
$\alpha(\m{zip}) = 3$. The reason is that with the modified assignment
$\m{zip}~f~(\m{cons}~x~\mi{xs})~(\m{cons}~y~\mi{ys}) \gtB f~x~y$ does not hold, 
as it requires the impossible multiset comparison
$\{ f, \m{cons}~x~\mi{xs}, \m{cons}~y~\mi{ys} \} \gtB^\mul \{ f~x, y \}$.
On the other hand, with $\alpha(\m{zip}) = 0$ 
the requirement changes to the possible one
$\{ \m{zip}~f~(\m{cons}~x~\mi{xs}), \m{cons}~y~\mi{ys} \} \gtB^\mul \{ f~x, y \}$.
Applying reification to the left-hand side is a key to handle the
variable application $f~x$.
\end{example}

Let us take a look at limitations of AMPOs, too.

\begin{example}
\label{ex:limitation}
The one-rule ATRS
\(
\m{f}~(\m{f}~x) \to \m{f}~(\m{g}~(\m{f}~x))
\)
is beyond the realm of our order, because the opposite orientation
$\m{f}~(\m{g}~(\m{f}~x)) \gtM \m{f}~(\m{f}~x)$ follows from
$\m{g}~x \gtM x$.  
Similarly, the one-rule ATRS $\m{a} \to \m{f}~\m{f}$
cannot be handled by our order: If $\alpha(\m{f}) = 0$ then
$\m{f}~\m{f} \gtM \m{a}$ holds.  
Otherwise,
we do not have $\m{a} \gtM \m{f}$.  
\end{example}

These limitations are essentially due to the inductive definition of the base
order.  The next example shows a limitation on precedences for AMPO.

\begin{example}[continued from \cref{ex:map}]
\label{ex:maximality-matters}
If we add the rule $\m{mapid} \to \m{map}~\m{id}$ to the ATRS, the
termination is no longer shown by AMPO:  To orient it, the
precedence must satisfy $\m{mapid} \succ \APP$,
which contradicts 
the maximality of $\APP$.
\end{example}

We have evaluated \cref{cor:termination} on 202 ATRSs in the
Termination Problem Database (TPDB).  Suitable arity assignments and
precedences are searched by the SMT solver Z3~\cite{MB08};
see~\cite{CGST12,ZHM09} for the encoding techniques.\footnote{%
The full experimental data, including those of LPO, EPO, and uncurrying
(see \cref{sec:conclusion}),
are available at \url{https://www.jaist.ac.jp/~hirokawa/26ijcar/}.}
Out of the 202 ATRSs, 143 are left head variable free, and 24 are
shown terminating by AMPO, while 2 are shown by MPO.
Note that the termination tools AProVE~\cite{GAB+17} and NaTT~\cite{YKS14}
prove termination of 165 and 148 ATRSs, respectively. Both subsume
the results of AMPOs.

\section{Conclusion}
\label{sec:conclusion}

Although there are
many variants~\cite{BJR15,JR07,KvR08,LB98,Toy04} of the recursive/multiset
path order for higher-order rewriting, here we compare ours with
ones having close relationship.
The lambda-free recursive path order by Blanchette et al.~\cite{BWW17} 
has been developed for automated theorem
proving.  It satisfies almost all properties for being a reduction order,
but it is not closed under contexts.  Recently, Bentkamp introduced the
embedding path order (EPO)~\cite{B21}.  As a family of MPO, this is the
first reduction order.  It can only orient \emph{non-duplicating} rules
$\ell \to r$ in which no variable occurs more in $r$ than in $\ell$.  Thus,
termination of higher-order recursive functions like map and
filter (\cref{ex:map,ex:filter}) cannot be shown by the order.  It should
be noted that the lambda-free RPO and the EPO satisfies totality on ground
terms, which guarantees refutational completeness of ordered completion~\cite{BDP89}
as well as of superposition calculus~\cite{BG94}.  Unfortunately, our order does not enjoy
the property: if $\alpha(\m{f}) = 1$, $\alpha(\m{a}) = 0$ and
$\m{f} \succ \m{a}$ then $\m{f}\,\m{a}$ and $\m{f}$ are incomparable with
$\gtB$ and also with $\gtM$.
Jouannaud and Rubio~\cite{JR07} developed the higher-order recursive
path order on \emph{typed} terms.  Due to differences in term structure it
is hard to relate their work to ours, but it is worth noting that
Buchholz's method~\cite{B95} is an untyped version of their well-foundedness proof.

Uncurrying~\cite{HMZ13,ST11} transforms left head variable free ATRSs into
functional TRSs in a termination reflecting manner.  
For example, the termination of the map ATRS in \cref{ex:map} follows from that of
the functional TRS consisting of the uncurried rules
\begin{alignat*}{4}
\m{append}_2(\m{cons}_2(x, \mi{xs}), \mi{ys}) &\to \m{cons}_2(x, \m{append}_2(\mi{xs}, \mi{ys}))
&
\mathllap{\m{append}_2(\m{nil}, \mi{ys})} &\to \mi{ys}
\\
\m{map}_2(f,\m{cons}_2(x,\mi{xs}))  & \to \m{cons}_2({\circ}(f,x),\m{map}_2(f,\mi{xs}))
&
\m{map}_2(f,\m{nil}) &\to \m{nil}
\\
\m{map}_2(f, \m{append}_2(\mi{xs}, \mi{ys})) & \to \m{append}_2(\m{map}_2(f, \mi{xs}), \m{map}_2(f, \mi{ys}))
&
\m{id}_1(x) & \to x
\\
\m{map}_2(\m{id},\mi{xs}) &\to \m{id}_1(\mi{xs})
&
\m{map}_1(\m{id}) &\to \m{id}
\end{alignat*}
and the uncurrying rules:
\begin{align*}
{\circ}(\m{append},\mi{xs}) &\to \m{append}_1(\mi{xs})
&
{\circ}(\m{append}_1(\mi{xs}), \mi{ys}) &\to \m{append}_2(\mi{xs},\mi{ys})
\\
{\circ}(\m{map},f) &\to \m{map}_1(f)
&
{\circ}(\m{map}_1(f),\mi{xs}) &\to \m{map}_2(f,\mi{xs})
\\
{\circ}(\m{cons},x) &\to \m{cons}_1(x)
&
{\circ}(\m{cons}_1(x), \mi{xs}) &\to \m{cons}_2(x,\mi{xs})
\\
{\circ}(\m{id},x) &\to \m{id}_1(x)
\end{align*}
The termination of the above TRS is shown by the multiset path order with
the precedence
\(
{\circ} \sim \m{map}
\succ \m{append}
\succ \m{id}
\succ \m{cons}
\succ \m{nil}
\) and $f \sim f_i$ for all function symbol $f$ and its indexed copies $f_i$.
Here ${\circ}$ being a maximum is used for orienting the uncurrying
rules.  If we identify $\circ$ with $\APP$ and each indexed copy
$f_i$ with the original symbol $f$, the precedence is the same as the one
in \cref{ex:map}.  
We anticipate that the applicative multiset path order can be simulated by a combination of
generalized uncurrying~\cite{ST11} and the original multiset path order.

The lexicographic path order (LPO)~\cite{KL80} is a famous variant of MPO
which compares arguments lexicographically.  Of course, such a version
could be considered in a similar fashion.  However,
if we adopt \emph{left-to-right}
lexicographic comparison, \cref{lem:reification} (a key lemma) fails to
hold.  We expect
that \emph{right-to-left} lexicographic comparison is a viable alternative.

Last but not least, 
because of the construction used in \cref{def:gtM} (by considering
all grounding substitutions), it is obscure if the relation $s \gtM t$ is
decidable for a fixed pair of precedence and arity assignment.  Developing
a decision procedure or finding a generalization of \cref{thm:hvf} would
enable us to use completion~\cite{KB70} and rewriting induction~\cite{R90}
more effectively.

\begin{credits}
\subsubsection{\ackname}
We are grateful to the reviewers for their valuable suggestions.
This research was supported by JSPS KAKENHI Grant Numbers JP22K11900 and
JP25KJ1363.
\subsubsection{\discintname}
The authors have no competing interests to declare that are
relevant to the content of this article.
\end{credits}

\bibliographystyle{splncs04}
\bibliography{references}

\appendix
\section{Omitted Proofs}
\label{sec:appendix}

Recall the statement of \cref{lem:sensibility}:
If $s \gtB t$ then $\Var(s) \supseteq  \Var(t)$.

\begin{proof}[of \cref{lem:sensibility}]
Suppose $s = f\,\aseq[m]{s} \gtB t$.  We show $\Var(s) \supseteq  \Var(t)$ by
induction on the derivation of $s \gtB t$.
\begin{enumerate}[1.]
\item
If $s \gtB t$ is derived by \AMPO{1} then $s_i = t$ or $s_i \gtB t$ for
some $i$.  In the former case $\Var(s_i) = \Var(t)$ holds.  In the latter
case the induction hypothesis yields $\Var(s_i) \supseteq \Var(t)$.  So
in either case $\Var(s) \supseteq \Var(s_i) \supseteq \Var(t)$ follows.
\item
If $s \gtB t$ is derived by \AMPO{2} then $t$ is of the form $g\,\aseq{t}$
with $g \in \Sigma$ and $s \gtB t_j$ for all $j$.  By the induction
hypothesis $\Var(s) \supseteq \Var(t_j)$ for all $j$.  Thus, the claim
holds.
\item
If $s \gtB t$ is derived by \AMPO{3} then $t$ is of the form $g\,\aseq{t}$
with $g \in \Sigma$ and $\{\seq[m]{s}\} \gtB^\mul \{\seq{n}\}$.  For every
$j$ there exists an $i$ such that $s_i = t_j$ or $s_i \gtB t_j$ holds.  
For the former case $\Var(s_i) = \Var(t_j)$ holds.  For the latter case
$\Var(s_i) \supseteq \Var(t_j)$ holds by the induction hypothesis.  
So in either case $\Var(s) \supseteq \Var(t_j)$ follows.  Thus, the claim
follows. 
\item
If $s \gtB t$ is derived by \AMPO{4} then $s \gtB \reify{t}$.  By the
induction hypothesis $\Var(s) \supseteq \Var(\reify{t})$ holds.  As
$\Var(\reify{t}) = \Var(t)$, the claim follows.
\item
If $s \gtB t$ is derived by \AMPO{5} then $\reify{s} \gtB t$.  By the
induction hypothesis $\Var(\reify{s}) \supseteq \Var(t)$ holds.  As
$\Var(s) = \Var(\reify{s})$, the claim follows.
\qedhere
\end{enumerate}
\end{proof}

Recall the statement of \cref{lem:head-preservation}: Suppose $\sigma$
is head-preserving for $s$ and $t$.  If $s \gtB t$ then 
$s\sigma \gtB t\sigma$.
 
\begin{proof}[of \cref{lem:head-preservation}]
Suppose $s = f\,\aseq[m]{s} \gtB t$.
We show $s\sigma \gtB t\sigma$ by induction on the derivation of 
$s \gtB t$.  Note that $\sigma$ is head-preserving for all subterms in $s$
and $t$.
\begin{enumerate}[1.]
\item
If $s \gtB t$ is derived by \AMPO{1} then $s_i = t$ or $s_i \gtB t$ holds
for some $i$.  In the former case $s_i\sigma = t\sigma$ holds trivially. In
the latter case $s_i\sigma \gtB t\sigma$ follows from the induction
hypothesis.  Therefore, $s\sigma \gtB t\sigma$ is derived by \AMPO{1}.
\item
If $s \gtB t$ is derived by \AMPO{2} then $t$ is a full term 
$g\,\aseq{t}$ with $g \in \Sigma$ and $f \succ g$, and
$s \gtB t_j$ for all $j$.  By the induction hypothesis
$s\sigma \gtB t_j\sigma$ for all $j$.  Therefore, $s\sigma \gtB t\sigma$ is
derived by \AMPO{2}.
\item
If $s \gtB t$ is derived by \AMPO{3} then the proof is analogous to the
last case.
\item
If $s \gtB t$ is derived by \AMPO{4} then $t$ is reifiable and not full, and
$s \gtB \reify{t}$.  As $\sigma$ is head-preserving for $t$, the instance
$t\sigma$ is reifiable and not full.  So the claim is immediate from
the induction hypothesis.
\item
If $s \gtB t$ is derived by \AMPO{5} then the proof is analogous to the
last case.
\qedhere
\end{enumerate}
\end{proof}

In the remaining part of the appendix we show that orientability of a
finite LATRS with AMPO is decidable.  
We use the notion of applicative arity~\cite{HMZ13}.
Let $\cR$ be a finite ATRS.  The \emph{applicative arity} $\aaR(f)$ of a
constant $f \in \Sigma$ is defined as follows~\cite{HMZ13}:
\[
\aaR(f) = \max\{ n \mid f\,\aseq{t} \in S_\cR \}
\]
Here $S_\cR$ consists of all subterms occurring in $\cR$
and $\max\varnothing$ is assumed to be $0$.
We write $T_\cR$ for the set
\(
S_\cR \cup \{ \reify{t} \mid \text{$t \in S_\cR$ and $t$ is reifiable} \}
\).
Given a rule $\ell \to r \in \cR$, the derivation of $\ell \gtB r$ uses
terms in $T_\cR$ only.  We define the arity assignment $\beta$ as
$\beta(f) = \min\{\alpha(f), \aaR(f)+1\}$.  

\begin{lemma}
\label{lem:fullness}
A term in $T_\cR$ is full with respect to $\alpha$ if and only if it is
full with respect to $\beta$.
\end{lemma}
\begin{proof}
Consider a term $t$ in $T_\cR$ of the form $f\,\aseq{t}$ with $f \in \Sigma$.
If $t$ is full with respect to $\alpha$ then 
$\alpha(f) = n \leqslant \aaR(f)$.  So we obtain 
$\beta(f) = \alpha(f)$.  Thus, $t$ is full with respect to $\beta$.
Conversely, if $t$ is full with respect to $\beta$ then
$\beta(f) = n < n + 1 \leqslant \aaR(f) + 1$, from which
$\beta(f) \neq \aaR(f)$ follows.  Thus, $n = \beta(f) = \alpha(f)$, and
therefore $t$ is full with respect to $\alpha$.
\end{proof}

Therefore, in any derivation of $\ell \gtB r$ for $\ell \to r \in \cR$
fullness with respect to $\alpha$ is indistinguishable from fullness with
respect to $\beta$.  Since the only role of an arity assignment in
\cref{def:ampo} is to determine fullness of terms, 
$\ell \gtB r$ with $\alpha$ carries over to $\ell \gtB r$ with $\beta$.

\begin{lemma}
\label{lem:orientability}
Let $\alpha$ be an arity assignment and let $s, t \in T_\cR$.
\begin{enumerate}[(1)]
\item
If $s \gtB t$ with respect to $\alpha$ then $s \gtB t$ with respect to
$\beta$.
\item
If $s \gtE t$ with respect to $\alpha$ then $s \gtE t$ with respect to
$\beta$.
\item
If $s >_\ampo t$ with respect to $\alpha$ then $s >_\ampo t$ with respect to
$\beta$.
\end{enumerate}
Hence, if $\cR \subseteq {>_\ampo}$ with respect to $\alpha$ then
$\cR \subseteq {>_\ampo}$ with respect to $\beta$.
\end{lemma}
\begin{proof}
The first claim is shown by straightforward induction on the derivation of
$s \gtB t$ with using \cref{lem:fullness}.  The second and third claims are
immediate from the first and the second, respectively.
\end{proof}

\cref{lem:orientability} states that, when finding an arity assignment
$\alpha$ and a precedence $>$ that satisfy $\cR \subseteq {>_\ampo}$, one
may assume $\alpha(f) \leqslant \aaR(f) + 1$ for all $f \in \Sigma$.
So existence of a suitable arity assignment and precedence can be found
by enumeration.

\begin{theorem}
\label{thm:decidability}
Given a finite LATRS $\cR$, whether $\cR \subseteq {>_\ampo}$ holds for
some arity assignment and precedence is decidable.
\end{theorem}

\end{document}